\documentclass[english,reprint, longbibliography, superscriptaddress, breaklinks=true, showkeys, showpacs=false, nofootinbib]{revtex4}
\usepackage[T1]{fontenc}
\usepackage[latin9]{inputenc}
\usepackage{color}
\usepackage{babel}
\usepackage{amsmath}
\usepackage{amsthm}
\usepackage{amsfonts}
\usepackage{amssymb}
\usepackage{graphicx}
\usepackage{subfig}
\usepackage{physics}

\newtheorem{teo}{Theorem}

\usepackage[unicode=true,pdfusetitle,
 bookmarks=true,bookmarksnumbered=false,bookmarksopen=false,
 breaklinks=true,pdfborder={0 0 0},backref=false,colorlinks=true]
 {hyperref} 
\usepackage{breakurl}

\makeatletter
\@ifundefined{textcolor}{}
{%
 \definecolor{BLACK}{gray}{0}
 \definecolor{WHITE}{gray}{1}
 \definecolor{RED}{rgb}{1,0,0}
 \definecolor{GREEN}{rgb}{0,1,0}
 \definecolor{BLUE}{rgb}{0,0,1}
 \definecolor{CYAN}{cmyk}{1,0,0,0}
 \definecolor{MAGENTA}{cmyk}{0,1,0,0}
 \definecolor{YELLOW}{cmyk}{0,0,1,0}
}

\pdfoutput=1
\hypersetup{colorlinks=true,citecolor=blue,linkcolor=cyan,urlcolor=blue,filecolor= green, breaklinks=true}
\usepackage{url}
\usepackage{breakurl}

\makeatother

\begin{document}

\title{An uncertainty view on complementarity and a complementarity view on uncertainty}

\author{Marcos L. W. Basso}
\email{marcoslwbasso@mail.ufsm.br}
\address{Departamento de F\'isica, Centro de Ci\^encias Naturais e Exatas, Universidade Federal de Santa Maria, Avenida Roraima 1000, Santa Maria, Rio Grande do Sul, 97105-900, Brazil}

\author{Jonas Maziero}
\email{jonas.maziero@ufsm.br}
\address{Departamento de F\'isica, Centro de Ci\^encias Naturais e Exatas, Universidade Federal de Santa Maria, Avenida Roraima 1000, Santa Maria, Rio Grande do Sul, 97105-900, Brazil}

\selectlanguage{english}%

\begin{abstract}
Since the uncertainty about an observable of a system prepared in a quantum state is usually described by its variance, when the state is mixed, the variance is a hybrid of quantum and classical uncertainties. Besides that, complementarity relations are saturated only for pure, single-quanton, quantum states. For mixed states, the wave-particle quantifiers never saturate the complementarity relation and can even reach zero for a maximally mixed state. So, to fully characterize a quanton it is not sufficient to consider its wave-particle aspect; one has also to regard its correlations with other systems. In this paper, we discuss the relation between quantum correlations and local classical uncertainty measures, as well as the relation between quantum coherence and quantum uncertainty quantifiers. We obtain a complete complementarity relation for quantum uncertainty, classical uncertainty, and predictability. The total quantum uncertainty of a $d$-paths interferometer is shown to be equivalent to the Wigner-Yanase coherence and the corresponding classical uncertainty is shown to be an entanglement monotone. The duality between complementarity and uncertainty is used to derive quantum correlations measures that complete the complementarity relations for $l_{1}$-norm and $l_{2}$-norm coherences. Besides, we show that Brukner-Zeilinger's invariant information quantifies both the wave and particle characters of a quanton and we obtain a sum uncertainty relation for the generalized Gell Mann's matrices.
\end{abstract}


\maketitle

\section{Introduction}
\label{sec:intro}
Quantum phenomena are manifestly unpredictable. While classical uncertainty arises from ignorance, quantum uncertainty is intrinsic. Even for pure quantum states that represent the maximal knowledge that one could have about quantum states, we can only make probabilistic predictions. The situation gets even worse when we consider two incompatible observables of a system. This is captured by the uncertainty relations, like the Heisenberg-Robertson uncertainty relation \cite{Robertson}, which is represented by the expression
\begin{equation}
    \mathcal{V}(\rho,A)\mathcal{V}(\rho,B) \ge \frac{1}{4} \abs{\Tr(\rho[A,B])}^2, \label{eq:hun}
\end{equation}
where $\mathcal{V}(\rho,A) = \Tr \rho A^2 - (\Tr \rho A)^2$ is the variance of the observable $A$ in the quantum state $\rho$, and $\mathcal{V}(\rho,B)$ is defined similarly. The existence of incompatible observables in quantum mechanics is somewhat related to quantum coherence, a kind of quantum superposition \cite{theurer}. However, in experiments, most quantum states are mixed, which means that a part of the unpredictability is actually classical. Since the uncertainty of an observable in a quantum state is usually described using the variance, when the states are mixed, the variance is a hybrid of quantum and classical uncertainties. In Ref. \cite{Luo}, Luo proposed a decomposition of the variance into classical and quantum parts. As pointed out by Luo, the key observation is that the Wigner-Yanase skew information \cite{wigner} can be interpreted as a measure of quantum uncertainty and the classical uncertainty can be captured by the difference between the total variance and the quantum uncertainty.
Later, the same author also established a different uncertainty  relation,  which  is  stronger  than  Eq. (\ref{eq:hun}), by taking into account only the quantum uncertainties \cite{Luu}. More recently, the same decomposition was done,  in Ref. \cite{Korzekwa}, for entropic uncertainty relations. In addition, it is worth mentioning that, in this manuscript, we deal with the classical local uncertainty that arises from the entanglement of the system $A$ with another system $B$ such that the global system state is pure. Therefore, the mixture of the quantum system A is due to the fact that we are ignoring system B. 
This is known as an improper mixture. In contrast, proper mixtures arise from the ignorance that we have about the preparation of the system, and this is known as the ignorance interpretation \cite{Paul}. However, in Refs. \cite{Luo, Luu} Luo did not make distinction between these two types of mixtures.

Another intriguing aspect of quantum mechanics is the wave-particle duality \cite{Bohr}. This characteristic is generally captured, in a qualitative way, by Bohr's complementarity principle. It states that quantons \cite{Leblond} have characteristics that are equally real, but mutually exclusive. It is known that, in a two-way interferometer, such as the Mach-Zehnder interferometer or the double-slit interferometer, the wave aspect is characterized by interference fringes, meanwhile the particle nature is given by the which-way information of the path along the interferometer, so that the complete knowledge of the path destroys the interference pattern and vice-versa. A quantitative version of the wave-particle duality was first investigated by Wooters and Zurek \cite{Wootters}, and later this work was extended by Englert, who derived a wave-particle duality relation \cite{Engle} between distinguishability and visibility as measures of the particle and wave aspects, respectively. In addition, Englert and Bergou \cite{Bethold} pointed out the possible connection between the distinguishability and quantum correlations and even conjectured that an entanglement measure was hidden in the measure of distinguishability. However, there is another way wherein the wave-particle duality has been captured, without introducing path-detecting devices. Greenberger and Yasin \cite{Yasin}, considering a two-beam interferometer, in which the intensity of each beam was not necessarily the same, defined a measure of path information, called predictability. This line of reasoning resulted in a different kind of wave-particle relation
\begin{equation}
    P^2 + V^2 \le 1, \label{eq:cr1}
\end{equation}
where $P$ is the predictability and $V$ is the visibility of the interference pattern. Several important steps have been taken towards the quantification of the wave-particle duality by many authors, such as D\"urr \cite{dur}  and Englert et al. \cite{englert}, who established minimal and reasonable conditions that any visibility and predictability measure should satisfy. As well, with the rapid development of the field of Quantum Information Science, it was suggested that quantum coherence \cite{Baumgratz} would be a good generalization of the visibility measure \cite{Bera, Bagan, Tabish, Mishra}. Meanwhile, predictability is a measure of the knowledge about the quantum level wherein a quanton \cite{Leblond} can be found. These levels can represent, besides the paths on a Mach-Zehnder interferometer, energy levels of an atom \cite{Xu} or, more generally, population levels \cite{Vlatko}. So far, many lines of reasoning were taken for quantifying the wave-particle properties of a quantum system \cite{Angelo, Coles, Hillery, Qureshi, Maziero, Lu}. It's noteworthy that, in \cite{Baumgratz, streltsov, theurer} the authors put forward a resource theory for quantum coherence, wherein is established minimum conditions that any measures of coherence must satisfy. However, it is worth to emphasize that the criteria for measures of coherence are not the same as the criteria for measures of visibility. So much so that Hilbert-Schmidt's (or $l_2$-norm) quantum coherence is considered a good measure of visibility, as shown in Ref. \cite{Maziero}, while it does not meet all the criteria for a good measure of quantum coherence, since it does not satisfy the condition of not increasing under incoherent operations.

Complementarity relations like the one in Eq. (\ref{eq:cr1}) are saturated only for pure, single-particle, quantum states. For mixed states, the left hand side is always less than one and can even reach zero for a maximally mixed state. Hence no information about the wave and particle aspects of the system can be obtained. As noticed by Jakob and Bergou \cite{Janos}, this lack of knowledge about the system is due to quantum entanglement \cite{Bruss, Horodecki}. This means that the information is being shared with another system and this kind of quantum correlation can be seen as responsible for the loss of purity of each subsystem such that, for pure maximally entangled states, it is not possible to obtain information about the local properties of the subsystems. So, to fully characterize a quanton, it is not enough to consider its wave-particle aspect, one has also to account for its correlations with other systems. Therefore, triality relations like the ones in \cite{Janos, Jakob, Marcos}, which completely quantifies the complementarity aspect of a quantum system. These triality relations are also known as complete complementarity relations, since in Ref. \cite{Eberly} the authors interpreted this equality as completing the duality relation given by Eq. (\ref{eq:cr1}), thus turning the inequality into an equality.

Hence, in the context of complementarity relations, we discuss, in this work, the relationship between Luo's criteria for quantum and classical uncertainties and D\"urr-Englert et al.'s criteria for wave-particle duality quantifiers together with the criteria for entanglement measures for global pure states, from which it follows naturally that quantum entanglement gives rise to local classical uncertainties, provided that the quanton is part of a pure bipartite quantum system, while quantum coherence gives rise to quantum uncertainties. In addition, we show that the quantum uncertainty of all $d$ paths is equivalent to the Wigner-Yanase quantum coherence \cite{yu_Cwy}, whereas the classical uncertainty can be taken as a entanglement monotone for bipartite pure cases. These results strengthen the idea stated in Ref. \cite{Sun}, that quantum coherence and quantum uncertainty are dual viewpoints of the same quantum substrate, and extend this statement by relating classical uncertainty with quantum correlations. Finally, by exploring the relation between complementarity and uncertainty, we obtain quantum entanglement measures completing the $l_1$-norm and $l_2$-norm complementarity relations reported in Ref. \cite{Maziero}. Such measures, as well, can be taken as measures of classical uncertainty. However, it is worth pointing out that, for the general case of a bipartite mixed state, the measures of classical uncertainty cannot be taken as measures of entanglement. 

The remainder of our paper is organized as follows: in Sec. \ref{sec:relation}, we discuss the relationship between Luo's criteria for quantum and classical uncertainties and D\"urr-Englert et al.'s criteria for wave-particle duality quantifiers. In Sec. \ref{sec:comp}, we obtain a complete complementarity relation (CCR) involving quantum and classical uncertainties and predictability. In addition, we show that the quantum uncertainty of all $d$-paths is equivalent to the Wigner-Yanase quantum coherence, meanwhile the classical uncertainty can be taken as a correlation quantifier. In Sec. \ref{sec:unce}, by exploring the duality of complementarity and uncertainty, we obtain quantum correlation measures that complete the $l_1$-norm and $l_2$-norm complementarity relations. For last, our conclusions are given in Sec. \ref{sec:conc}.

\section{Relating criteria for uncertainties and criteria for complementarity quantifiers}
\label{sec:relation}
In the formalism of Quantum Mechanics \cite{Messiah}, when the system is in the state $\rho$,  the uncertainty of an observable $A$, here restricted to be an Hermitian operator, like the path of a multi-slit interferometer, is given by the variance 
\begin{equation}
    \mathcal{V}(\rho,A) = \Tr \rho A^2_{0}  = \Tr \rho A^2 - (\Tr \rho A)^2,
\end{equation}
where $A_0 = A - \Tr \rho A$. Since, in general, $\rho$ describes a mixed state, the variance $\mathcal{V}(\rho,A)$ quantifies both quantum and classical uncertainties. Luo \cite{Luo} proposed to split the variance in its quantum and classical parts
\begin{equation}
     \mathcal{V}(\rho,A) =  \mathcal{Q}(\rho,A) +  \mathcal{C}(\rho,A),
\end{equation}
where $\mathcal{Q}(\rho,A)$ and $\mathcal{C}(\rho,A)$ correspond to the quantum and classical uncertainties, respectively. Luo also established a set of reasonable conditions that any measure of quantum and classical uncertainty should satisfy. For $\mathcal{Q}(\rho,A)$, these required properties can be stated as follows:
\begin{itemize}
    \item[Q.1] If $\rho$ is pure, then $\mathcal{V}(\rho,A) = \mathcal{Q}(\rho, A)$ and $\mathcal{C}(\rho, A) = 0$, because there is no classical mixing and all uncertainties are intrinsically quantum.
    \item[Q.2] If $[\rho, A] = 0$, both are diagonal in the same basis and $\rho$ and $A$ behaves like classical variables. Hence, all uncertainties are classical, i.e., $\mathcal{Q}(\rho, A) = 0$ and $\mathcal{V}(\rho,A) = \mathcal{C}(\rho, A)$. 
    \item[Q.3] $\mathcal{Q}(\rho,A)$ must be convex in $\rho$, once classical mixing does not increase quantum uncertainty, i.e., $\mathcal{Q}(\sum_i \lambda_i \rho_i, A) \le \sum_i \lambda_i \mathcal{Q}(\rho_i, A)$ with $\sum_i \lambda_i = 1$, $\lambda_i \in [0,1]$, and $\rho_i$ are valid quantum states.
\end{itemize}
Meanwhile, for $\mathcal{C}(\rho,A)$:
\begin{itemize}
    \item[C.1] The same as Q.1.
    \item[C.2] The same as Q.2.
    \item[C.3] $\mathcal{C}(\rho,A)$ must be concave in $\rho$, once classical mixing increases classical uncertainty, i.e., $\mathcal{C}(\sum_i \lambda_i \rho_i, A) \ge \sum_i \lambda_i \mathcal{C}(\rho_i, A)$ with $\sum_i \lambda_i = 1$, $\lambda_i \in [0,1]$, and $\rho_i$ are well defined quantum states.
\end{itemize}

Also, D\"urr \cite{dur} and Englert et al. \cite{englert} established criteria that can be taken as a standard for checking for the reliability of newly defined predictability measures $P(\rho)$ and interference pattern visibility quantifiers $V(\rho)$. For $P$, these required properties can be stated as follows:
\begin{itemize}
\item[P.1] $P$ must be a continuous function of the diagonal elements of the density matrix.
\item[P.2] $P$ must be invariant under permutations of the states indexes.
\item[P.3] If $\rho_{jj}=1$ for some $j$, then $P$ must reach its maximum value.
\item[P.4] If $\{\rho_{jj}=1/d\}_{j=1}^{d}$, then $P$ must reach its minimum value.
\item[P.5] If $\rho_{jj}>\rho_{kk}$ for some $(j,k)$, the value of $P$ cannot be increased by setting $\rho_{jj}\rightarrow\rho_{jj}-\epsilon$ and $\rho_{kk}\rightarrow\rho_{kk}+\epsilon$, for $\epsilon\in\mathbb{R}_{+}$ and $\epsilon\ll1$.
\item[P.6] $P$ must be a convex function, i.e., $P(\sum_i \lambda_i \rho_i)\le \sum_i \lambda_i P(\rho_i)$, with $\sum_i \lambda_i = 1$, $\lambda_i \in [0,1]$ and for $\rho_i$ being valid density matrices.  
\end{itemize}

Meanwhile, for any measure of the wave aspect $V$ of a quanton:
\begin{itemize}
\item[V.1] $V$ must be a continuous function of the elements of the density matrix.
\item[V.2] $V$ must be invariant under permutations of the states' indexes.
\item[V.3] If $\rho_{jj}=1$ for some $j$, then $V$ must reach its minimum value.
\item[V.4] If $\rho$ is a pure state and $\{\rho_{jj}=1/d\}_{j=1}^{d}$, then $V$ must reach its maximum value.
\item[V.5] $V$ cannot be increased when decreasing $|\rho_{jk}|$ by an infinitesimal amount, for $j\ne k$.
\item[V.6] $V$ must be a convex function, i.e., $V(\sum_i \lambda_i \rho_i)\le \sum_i \lambda_i V(\rho_i)$ with $\sum_i \lambda_i = 1$, $\lambda_i \in [0,1]$, and $\rho_i$ are well defined density matrices.  
\end{itemize} 

In order to explore the relationship between the conditions for a quantum uncertainty measure and those for a visibility measure, let us restrict ourselves to the context of multi-slit interferometry, i.e, let us consider that the observable $A$ is the projection onto one of the $d$-paths of the interferometer: $A = \ketbra{j}$, for some path (state) label $j$. In the extreme case where $\rho$ is pure and $\rho_{jj} = 1/d\ \forall j$, the quantum uncertainty must be maximal $\mathcal{V}(\rho,A) = \mathcal{Q}(\rho, A) = \mathcal{Q}^{max}$, and the visibility also reaches its maximum value. Besides, there is no classical uncertainty $\mathcal{C}(\rho, A) = 0$.  In the other extreme case, when $[\rho, A] = 0$, $\rho$ is an incoherent state in the basis shared by $\rho$ and $A$, thus all uncertainties are classical and $V = \mathcal{Q}(\rho, A) = 0$, since there is no coherence in this basis. Besides, if the state of the quanton is known, we have that $\rho$ is pure, and $\rho_{jj} = 1$ for some state index $j$. Thus $V = \mathcal{Q}(\rho, \ketbra{j}) = \mathcal{C}(\rho, \ketbra{j}) = 0$ and the predictability reaches its maximum value. In addition, the visibility and the quantum uncertainty must be convex functions of $\rho$, since classical mixture does not increase the coherence of $\rho$ and its quantum uncertainty.

On the other hand, the relation between classical uncertainty and quantum correlation is more subtle. It is known that complementarity relations for wave-particle duality are saturated only for pure, single-quanton, quantum states. For a maximally incoherent state the wave and particle quantifiers can reach zero and no information about the wave and particle aspects of the system can be obtained. So, the information is being shared with other systems, and these correlations can be seen as responsible for the increase of entropy of the quanton \cite{Janos}. Thus, if the system $\rho$ is not correlated with other systems, then $\rho$ must be pure. In this case, the classical uncertainty is $\mathcal{C}(\rho, \ketbra{j}) = 0$. Moreover, when $[\rho, A] = 0$, $\rho$ is an incoherent state in the eigenbasis shared by $\rho$ and $A$, all uncertainties are classical. However, we can always purify $\rho$ and think of it as resulting from entanglement with another system \cite{nielsen}. For instance, we can consider $\rho$ entangled with the states of a path detector device such that $\rho$ is incoherent in the path basis. Therefore, in this case, $\mathcal{C}(\rho, \ketbra{j}{j}) \neq 0$ is a signature of quantum entanglement. In addition, maximally incoherent reduced states are used to classify multipartite pure entangled states as maximally entangled. Beyond that, it is known that, for multipartite quantum systems, any entanglement measure must be a convex function \cite{Bruss}. However, the condition (C.3) is related to the particular subsystem $\rho$. Hence, the classical mixture $\rho =  \sum_i \lambda_i \rho_i$ can be recast as the effect of local measurements, which is classified as a Local Operation and Classical Communication procedure (LOCC) \cite{Mark}. This fact is a direct consequence of the Neumark's theorem \cite{Mark}. Hence, any entanglement measure must be concave under classical mixtures. Finally, we could add that the quantum and classical uncertainties should be continuous functions of the density matrix elements and invariant under permutations of the paths (states) indexes.


\section{A complementarity view of uncertainty}
\label{sec:comp}
To introduce quantum uncertainty, Luo considered the following definition \cite{Luo}:
\begin{equation}
    \mathcal{Q}(\rho, A) := \mathcal{I}_{wy}(\rho,A) = -\frac{1}{2}\Tr ([\sqrt{\rho},A_0]^2),
\end{equation}
where $\mathcal{I}(\rho,A)$ is the skew information introduced by Wigner and Yanase, also known as the Wigner-Yanase entropy. As pointed out by Luo \cite{Luo}, their interpretation is that $\mathcal{Q}(\rho, A)$ quantifies the information contents of the quantum state $\rho$ with respect to observables not commuting with (i.e., skew to) the observable $A$. Besides that, $\mathcal{Q}(\rho, A)$ can also be regarded as quantifying the information of observables not commuting with $A$ in the state $\rho$. Because of Bohr's complementarity principle, we can further interpret $\mathcal{Q}(\rho, A)$ as some kind of uncertainty of $A$ itself in $\rho$. Hence, a natural definition of classical uncertainty is:
\begin{equation}
    \mathcal{C}(\rho,A) := \mathcal{V}(\rho, A) - \mathcal{Q}(\rho, A) = \Tr \sqrt{\rho}A_0 \sqrt{\rho}A_0.
\end{equation}

As before, let us consider the observable $A$ as the projection onto one of the paths (or slit) of the interferometer, i.e., $A = \ketbra{j}$, for some path (state) label $j$. It is worth pointed out that, in this work, we restrict ourselves to orthogonal measurements, since the path states can be considered orthogonal to each other. In this case, the quantum uncertainty of the path $j$ is given by
\begin{align}
    \mathcal{Q}(\rho, \ketbra{j}) & = -\frac{1}{2} \Tr ([\sqrt{\rho}, \ketbra{j}_0]^2)\\
    & = - \frac{1}{2}( \expval{\sqrt{\rho}}{j}^2 + \expval{\sqrt{\rho}}{j}^2 - \sum_k \bra{k} \sqrt{\rho} \ketbra{j} \sqrt{\rho} \ket{k} - \expval{\rho}{j})\\
    & = \expval{\rho}{j} - \expval{\sqrt{\rho}}{j}^2.
\end{align}
If $\rho$ is pure, then $\sqrt{\rho}=\rho$ and $\mathcal{Q}(\rho, A) = \expval{\rho}{j} - \expval{\rho}{j}^2$. For $\expval{\rho}{j} := \rho_{jj} = 1/d\ \forall j$, the quantum uncertainty reaches its maximum $\mathcal{Q}^{max} = (d - 1)/d^2$. On the other hand, if the path is known, i.e., $\rho_{kk} = 1$ for some path index $k$, then $\mathcal{Q}(\rho, A) = 0$, even for $k = j$. Now, if $[\rho, \ketbra{j}] = 0\ \forall j$, then $\rho$ is diagonal in the path basis and $\rho_{jj} = (\sqrt{\rho}_{jj})^2\ \forall j$, which implies that $\mathcal{Q}(\rho, A) = 0$. We can also define the quantum uncertainty of all $d$-paths:
\begin{align}
    \mathcal{U}_{q} & := \sum_j \mathcal{Q}(\rho, \ketbra{j}_0) = \sum_j(\expval{\rho}{j} - \expval{\sqrt{\rho}}{j}^2)\\
    & = \sum_j( \sum_k \bra{j}\sqrt{\rho}\ketbra{k}\sqrt{\rho}\ket{j} - \expval{\sqrt{\rho}}{j}^2) = \sum_{j,k} \abs{\bra{j}\sqrt{\rho}\ket{k}}^2 - \sum_j\expval{\sqrt{\rho}}{j}^2\\
    & = \sum_{j\neq k}\abs{\bra{j}\sqrt{\rho}\ket{k}}^2 = C_{wy}(\rho),
\end{align}
where $C_{wy}(\rho)$ is the Wigner-Yanase quantum coherence \cite{yu_Cwy}, which is a bona-fide measure of visibility, as we have shown in Ref. \cite{Maziero}. Besides, $\mathcal{U}_{q}$ also satisfies Luo's criteria for a quantum uncertainty. For the classical uncertainty of the path $j$, we have
\begin{align}
    \mathcal{C}(\rho, \ketbra{j})& =  \Tr \sqrt{\rho} \ketbra{j}_0 \sqrt{\rho} \ketbra{j}_0  = \expval{\sqrt{\rho}}{j}^2 - \expval{\rho}{j} \sum_k \bra{k}\rho \ketbra{j} \rho \ket{k} \\
    & = \expval{\sqrt{\rho}}{j}^2 - \expval{\rho}{j}^2.
\end{align}
If $\rho$ is pure, then $\mathcal{C}(\rho,A) = 0$. On the other hand, if $\rho$ is incoherent, then  $\mathcal{C}(\rho,A) \neq 0$, and for the extreme case $\rho = \sum_j \frac{1}{d} \ketbra{j}$, the classical uncertainty reaches its maximum $\mathcal{C}^{max} = (d - 1)/d^2$. Meanwhile, the classical uncertainty of all $d$-paths is given by 
\begin{align}
  \mathcal{U}_c = \sum_j \mathcal{C}(\rho, \ketbra{j}) = \sum_j(\expval{\sqrt{\rho}}{j}^2 - \expval{\rho}{j}^2).
\end{align}
Now, summing both uncertainties, we have
\begin{align}
    \mathcal{U}_q + \mathcal{U}_c & = \sum_{j,k} \abs{\bra{j}\rho\ket{k}}^2 - \sum_j \expval{\rho}{j}^2  = \Tr (\sqrt{\rho})^2 - \sum_j \expval{\rho}{j}^2 = 1 - \sum_j \expval{\rho}{j}^2\\
    & = S_l(\rho_{diag}), \label{eq:unc}
\end{align}
where $S_l(\rho)=1-\Tr\rho^{2}$ is the linear entropy. So, we can establish a complementarity relation between classical and quantum uncertainties: 
\begin{equation}
    \mathcal{U}_q + \mathcal{U}_c \le S_l^{max}.
\end{equation}
But it is possible to explore Eq. (\ref{eq:unc}) even further. Since for $d$ paths with probabilities $\rho_{11}, \rho_{22},\cdots, \rho_{dd}$, the lack of information about the $j$-th path is given by $\rho_{jj}(1 - \rho_{jj})$, the total lack of information about all the $d$-paths is given by $\sum_{j}\rho_{jj}(1 - \rho_{jj}) = 1 - \sum_j \rho^2_{jj}$, which is equal to $S_l(\rho_{diag}) = 1 - \Tr \rho^2_{diag}$ \cite{chen}. In other words, defining $\Pi_j := \ketbra{j}$ as the the projection onto the path (state) index $j$, the uncertainty of the path index $j$ is given by 
\begin{equation}
    \mathcal{V}(\rho,\Pi_j) = \Tr \rho \Pi_j^2 - (\Tr \rho \Pi_j)^2 = \rho_{jj} - \rho^2_{jj},
\end{equation} 
such that the total uncertainty of the paths is obtained by summing over $j$:
\begin{equation}
    \sum_j \mathcal{V}(\rho,\Pi_j) = 1 - \sum_j \rho^2_{jj}.
\end{equation}
Hence, as expected, $\mathcal{U}_q + \mathcal{U}_c = \sum_j \mathcal{V}(\rho,\Pi_j)$. Beyond that, Eq. (\ref{eq:unc}) can also be rewritten as a complete complementarity relation between uncertainty and predictability:
\begin{equation}
    \mathcal{U}_q + \mathcal{U}_c + P_l = S_l^{max}. \label{eq:unpl}
\end{equation}
Once $S_l(\rho_{diag})$ is measuring our total uncertainty (or ignorance) about the paths, we can interpret $P_l(\rho): = S_l^{max} - S_l(\rho_{diag})$ as measuring our capability of making a correct guess about the possible outcomes in the path basis, i.e., if our total uncertainty about the path decreases, our capability of making a correct guess has to increase.  Actually, $P_l(\rho)$ is a bona-fide predictability measure \cite{Maziero}. It is worthwhile emphasizing that predictive information is defined outside the realm of quantum information science as the difference between a prior and posterior entropy measures, and can be interpreted as the average information about the state contained in a
prediction \cite{Schneider}. The following theorem, which  relies largely on known results \cite{John, Bera, Dieguez}, is presented here as extensions regarding the facts that the coherences of $\rho$ give rise to quantum uncertainties and that the classical uncertainty is due to the possible quantum correlations with others systems, if we consider $\rho$ as part of a pure multipartite quantum system. 

\begin{teo}
 Let $\ket{\Psi}_{A,B}$ be a bipartite pure state of a quantum system. Then, quantum correlations give rise to local classical uncertainties and quantum coherences give rise to quantum uncertainties. Conversely, classical uncertainties are signatures of quantum correlations and quantum uncertainties are signatures of quantum coherences.
\end{teo}
\begin{proof}
Without loss of generality, in the context of $d$-slit interferometry, let $\ket{j}$ describe the state corresponding to the quanton taking the $j$-th path, the general state is given by $\ket{\psi}_A = \sum_j a_j \ket{j}$, where $a_j$ represents the probability amplitude of the quanton to take the $j$-th path, and $\{\ket{j}\}_{j = 1}^d$ can be regard as a orthonormal path basis. Consider now a path-detector which is capable of recording which path the quanton followed. This path detector is also a quantum object. In a von Neumann pre-measurement (\cite{John}, Chapters V and VI), the detector interacts with a quanton and gets entangled with it, i.e., $U(\ket{j} \otimes \ket{d_0}) \to \ket{j} \otimes \ket{d_j}$, where $\ket{d_0}$ is the initial detector state and $U$ represents the unitary evolution operator. Then, the state of the quanton and the detector is given by
\begin{equation}
    \ket{\Psi}_{A,B} = \sum_j a_j \ket{j} \otimes \ket{d_j},
\end{equation}
where $\ket{d_j}$ is the state of the path-detector corresponding to the quanton following the $j$-th path, and $\ket{\Psi}_{A,B}$ represents a bipartite pure quantum system.  Also, without loss of generality, we consider the detector states $\{\ket{d_j}\}_{j = 1}^d$ to be normalized, but not necessarily orthogonal. Now, if we consider only the state of the quanton, we have a mixed state described by \cite{John}
\begin{equation}
    \rho_A = \Tr_B(\ket{\Psi}_{A,B}\bra{\Psi}) = \sum_{j,k} a_j a^*_k \braket{d_k}{d_j} \ketbra{j}{k}.
\end{equation}
If the states of the detector are completely distinguishable, i.e., $\braket{d_k}{d_j} = \delta_{jk}$, then $\rho_A = \sum_j \abs{a_j}^2 \ketbra{j}{j}$ is an incoherent state, and $\rho_A$ commutes with any $\ketbra{j}{j}$. Hence, we have just classical uncertainty. On the other hand, if the detector does not couple with the quanton, then the bipartite quantum system is separable, and the state of the quanton is pure. Therefore, the uncertainty is only quantum. For last, if the detector states are not mutually orthogonal to each other, the off-diagonal elements of the reduced density matrix $\rho_A = \sum_{j,k} a_j a^*_k \braket{d_k}{d_j} \ketbra{j}{k}$ do not necessarily vanish. But, the coherence of the quanton will be certainly reduced in comparison with the pure state $\ket{\psi}_A = \sum_j a_j \ket{j}$ \cite{Bera}. Thus, part of the quantum uncertainty will be transformed into classical uncertainty, and we will have a mixture of both. It is easy to see this from Eq. (\ref{eq:unpl}), since $S_l^{max}$ is a constant and $P_l(\rho)$ is not affected by the states of the path detector. Conversely, if we have only quantum uncertainty, $\rho$ describes a pure state and there will be at least a superposition of two elements of the path basis, otherwise the path will be known, which contradicts the hypothesis that we have quantum uncertainty. At the other end, if we have only classical uncertainty, $\rho$ is incoherent in the path basis. However, it is always possible to purify $\rho$ by entangling it with another system. The trivial case where $\rho$ is a projector on one of the uni-dimensional sub-spaces of the path basis, then $\rho$ is pure and the path is known, which contradicts the hypothesis that we have classical uncertainty.
\end{proof}

Hence, if we accept that $\mathcal{U}_q = C_{wy}(\rho)$ is measuring the wave aspect and $P_l(\rho)$ is a measure of the particle aspect of the quanton, then $\mathcal{U}_c = \sum_j \Tr \sqrt{\rho} \ketbra{j}_0 \sqrt{\rho} \ketbra{j}_0$ can be considered a measure of entanglement of the quanton with other systems or degrees of freedom, provided that the global state is pure. 

\begin{teo}
 Let $\ket{\Psi}_{A,B} \in \mathcal{H}_A \otimes \mathcal{H}_B$ be the state of a bipartite pure quantum system, with $\rho_A = \Tr_B(\ket{\Psi}_{A,B}\bra{\Psi})$. Then $\mathcal{U}_c := \sum_j \mathcal{C}(\rho_A, \ketbra{j}) = \sum_j \Tr \sqrt{\rho_A} \ketbra{j}_0 \sqrt{\rho_A} \ketbra{j}_0$ is a entanglement monotone\footnote{Entanglement monotones are nonnegative functions whose value does not increase under local operations and classical communication (LOCC) \cite{Horodecki}.}, with $\sum_{j = 1}^d \ketbra{j} = I_{d \times d}$.
\end{teo}
\begin{proof}
\begin{itemize}
    \item If $\ket{\Psi}_{A,B}$ is separable, then $\rho_A = \Tr_B(\ket{\Psi}_{A,B}\bra{\Psi})$ is pure and $\mathcal{U}_c = \sum_j(\expval{\sqrt{\rho_A}}{j}^2 - \expval{\rho_A}{j}^2) = 0.$ Conversely, if $\mathcal{U}_c = 0$, then $\sqrt{\rho_A} = \rho_A$, which implies that $\rho_A$ is pure, and thus separable.
    \item $\mathcal{U}_c \ge 0$, once that $\mathcal{U}_c :=  \sum_j \Tr \sqrt{\rho_A} \ketbra{j}_0 \sqrt{\rho_A} \ketbra{j}_0  =  \sum_j \Tr \rho_A^{1/4} \ketbra{j}_0 \rho_A^{1/4} \rho_A^{1/4} \ketbra{j}_0 \rho_A^{1/4} = \sum_j \Tr X_j^{\dagger} X_j \ge 0$, where $X_j := \rho_A^{1/4} \ketbra{j}_0 \rho_A^{1/4}$.
    \item $\mathcal{U}_c$ is invariant under unitary local transformations. To see this, let $U_A \otimes U_B \ket{\Psi}_{A,B}$, where $U_A, U_B$ are unitary operators in $\mathcal{H}_A, \mathcal{H}_B$, respectively. Thus $\rho'_A = U_A \rho_A U_A^{\dagger}$. Following Ref. \cite{Shu}, it is enough to note that \begin{align}
    &\sum_j \mathcal{C}(U_A \rho_A U_A^{\dagger}, \ketbra{j}_0 ) = \sum_j \mathcal{C}( \rho_A, U_A^{\dagger} \ketbra{j}_0 U_A),
    \end{align}
    which implies that, for any local unitary transformation $U_A$, the set $\{U^{\dagger}_A \ketbra{j}{j}U^{\dagger} \}$ still is an orthonormal basis. Hence, $\mathcal{U}_c$ is invariant under unitary transformations. 
    \item $\mathcal{U}_c$ does not increase under classical mixing of $\rho_A$, which is a special type of LOCC \cite{Luo}. More generally, using the Schmidt decomposition $\ket{\Psi}_{A,B} = \sum_k \sqrt{\lambda_k}\ket{\phi_k}_A \otimes \ket{\psi_k}_B$, we can write $\sqrt{\rho_A} =  \sum_k \sqrt{\lambda_k}\ketbra{\phi_k}$. Thus
    $\mathcal{U}_c = \sum_j \Tr(\sum_k \sqrt{\lambda_k}\ketbra{\phi_k} \ketbra{j}_0 \sum_l \sqrt{\lambda_l}\ketbra{\phi_l}\ketbra{j}_0)$ is obviously invariant under the permutation of the Schmidt coefficients. Besides,
    \begin{equation} \frac{\partial \mathcal{U}_c}{\partial \lambda_m} = \sum_j \Tr \lambda^{-1/2}_m\ketbra{\phi_m} \ketbra{j}_0 \sum_l \sqrt{\lambda_l}\ketbra{\phi_l}\ketbra{j}_0, \nonumber 
    \end{equation}
    for $ m = 1,2$. Without loss of generality, if $\lambda_1 \ge \lambda_2$, then $\lambda^{-1/2}_1 \le \lambda^{-1/2}_2$ and
    \begin{equation}
        (\lambda_1 - \lambda_2)(\frac{\partial \mathcal{U}_c}{\partial \lambda_1} - \frac{\partial \mathcal{U}_c}{\partial \lambda_2}) \le 0.
    \end{equation}
Therefore $\mathcal{U}_c$ is monotonously decreasing under LOCC \cite{Mintert}.    
\end{itemize}
\end{proof}

It is worth pointing out the we are not claiming that $\mathcal{U}_c$ is an entanglement measure in the sense of concurrence, which applies to any bipartite quantum system (pure and mixed). We are claiming that $\mathcal{U}_c$ can be taken as a entanglement monotone, just as the von-Neumann entropy can be taken as a entanglement monotone, for a bipartite pure quantum system. Actually, in Ref. \cite{Leopoldo}, by a different route, we succeed in showing that $U_c$ is an entanglement monotone. 

\subsection{Brukner-Zeilinger invariant information and its relation with complementarity}
It is noteworthy the apparent similarity between the predictability $P_l(\rho) := S_l^{max} - S_l(\rho_{diag}) = Tr \rho^2_{diag} - 1/d$ and the Brukner-Zeilinger (BZ) invariant information $I_{BZ}(\rho) := Tr \rho^2 - 1/d$ \cite{Brukner}. However, there is a fundamental difference between these quantities: the predictability is basis-dependent while the BZ information is not. We can see this by considering a 2-level quantum system, whose state space is $\mathbb{C}^2$. Instead of using a Mach-Zehnder interferometer, we choose particles with spin-$1/2$ whose magnetic moment is measured using a Stern-Gerlach apparatus \cite{Sakurai}. The observables that we will consider are the components of the magnetic moment of these particles in the direction $z$, $S_z = (\hbar/2) \ketbra{z_+} - (\hbar/2)\ketbra{z_-}$, and in the direction $x$, $S_x = (\hbar/2) \ketbra{x_+} - (\hbar/2)\ketbra{x_-}$, where $\hbar$ is Planck's constant, $\ket{z_+} = [1 \ 0]^{\dagger}$, $\ket{z_-} = [0 \ 1]^{\dagger}$, and $\ket{x_{\pm}} = (\ket{z_+} \pm \ket{z_-})/\sqrt{2}$. In this case, $\{\ket{z_+},\ket{z_-}\}$ plays the role of the path basis, and $S_z$ can be the observable related to the path information. Meanwhile, $S_x$ is an incompatible observable of $S_z$. Now, if we consider an ensemble of particles prepared in the state
\begin{equation}
\rho = p_1 \ketbra{z_+} + p_2 \ketbra{x_+},
\label{eq:eqr}
\end{equation}
with $p_1 + p_2 = 1$, which possesses quantum uncertainty due to the incompatible observables $\ketbra{z_+}$ and $\ketbra{x_+}$, therefore possessing quantum coherence in the basis $\{\ket{z_+},\ket{z_-}\}$. We can see that $I_{BZ}(\rho)$ reaches a maximum when $p_1 = 0$ or $1$ (and $p_2 = 0$ or $1$) while $P_l(\rho)$ reaches a maximum when $p_1 = 1$ and a minimum when $p_1 = 0$, once the predictability is a measure related to the state (path) basis $\{\ket{z_+},\ket{z_-}\}$. In contrast, let's consider the state 
\begin{equation}
\sigma = p_1 \ketbra{z_+} + p_2 \ketbra{z_-},
\label{eq:eqs}
\end{equation}
with $p_1 + p_2 = 1$, which possesses only classical uncertainty. If we consider $p_1 = p_2 = 1/2$, then $I_{BZ}(\sigma) = 0$, whereas $I(\rho) = 1/4$, and $P_l(\rho) = P_l(\sigma)$. Thus, the BZ information can be lifted by quantum uncertainties while the behavior of $P_l$ is the same regardless of whether the nature of uncertainty is quantum, classical, or even a mixture of both. This stems from the fact that BZ information can be taken as a measure of the local properties of a quanton, i.e., its particle-wave nature. In Fig. \ref{fig:inv}, we plot the behavior of $I_{BZ}(\rho),I_{BZ}(\sigma), P_l:= P_{l}(\rho) = P_{l}(\sigma)$ as function of $p_1$. 
\begin{figure}[t]
\includegraphics[scale=0.6]{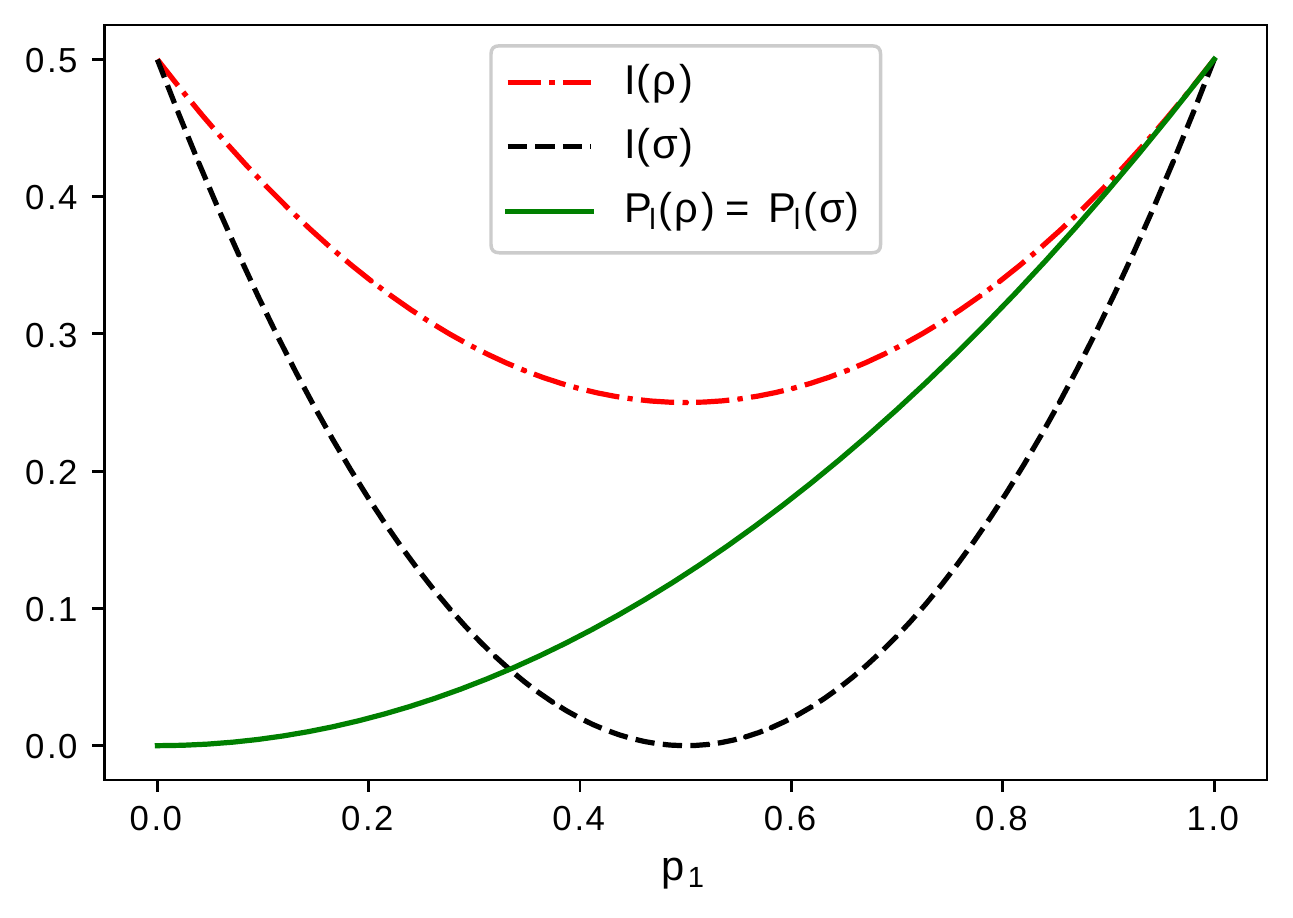}
\caption{(Color online) Brukner-Zeilinger invariant information and linear predictability of the states (\ref{eq:eqr}) and (\ref{eq:eqs}).}
\label{fig:inv}
\end{figure}

\begin{teo}
The Brukner-Zeilinger invariant information $I_{BZ}(\rho) := Tr \rho^2 - 1/d$ measures the local aspects of a quanton, i.e., its particle and wave aspects.
\end{teo}

\begin{proof}
The proof follows directly by the definition of BZ information:
\begin{align}
    I_{BZ}(\rho) & := Tr \rho^2 - 1/d = \sum_{j,k}\abs{\rho_{jk}}^2 - 1/d \\
    & = \sum_j \rho^2_{jj} - 1/d + \sum_{j \neq k} \abs{\rho_{jk}}^2 \\
    & = P_l(\rho) + C_{hs}(\rho),
\end{align}
where $C_{hs}(\rho) := \sum_{j \neq k} \abs{\rho_{jk}}^2$ is the Hilbert-Schmidt quantum coherence \cite{Jonas}, which is also a bona-fide measure of visibility \cite{Maziero}, that was already used in the works by Jakob and Bergou \cite{Janos, Jakob}.
\end{proof}

\section{An uncertainty view of complementarity}
\label{sec:unce}
Within this framework, we can interpret any complete complementarity relation in terms of uncertainty. For instance, for any quantum state $\rho$ of dimension $d$, the relative entropy of coherence is defined as  \cite{Baumgratz}
\begin{align}
    C_{re}(\rho) = \min_{\iota \in I} S_{vn}(\rho||\iota),
\end{align}
where $I$ is the set of all incoherent states, and $S_{vn}(\rho||\iota) = \Tr(\rho \ln \rho - \rho \ln \iota)$ is the relative entropy. The minimization procedure implies that $\iota = \rho_{diag} = \sum_{i} \rho_{ii} \ketbra{i}$, thus 
\begin{align}
    C_{re}(\rho) = S_{vn}(\rho_{diag}) - S_{vn}(\rho) \label{eq:cre}.
\end{align}
Once $C_{re}(\rho) \le S_{vn}(\rho_{diag})$, it is possible to obtain an incomplete complementarity relation from this inequality:
\begin{align}
    C_{re}(\rho) + P_{vn}(\rho) \le \ln d \label{eq:cr6},
\end{align}
with $P_{vn}(\rho) := \ln d - S_{vn}(\rho_{diag}) = \ln d + \sum_{i} \rho_{ii} \ln \rho_{ii}$ as a measure of the predictability, already defined in Refs. \cite{englert, Maziero}. Such measure is only possible to define because we can interpret the diagonal elements of $\rho$ as a probability distribution, which is a consequence of the properties of $\rho$ \cite{Maziero}. The complementarity relation (\ref{eq:cr6}) is incomplete due to the presence of correlations. However, if $\rho$ is a subsystem of a bipartite pure quantum system $\ket{\Psi}_{A,B}$, which allows us to take $S_{vn}(\rho)$ as a measure of entanglement of the subsystem A with B \cite{Mark}, so it is possible to interpret Eq. (\ref{eq:cre}) as a complete complementarity relation:
\begin{align}
         C_{re}(\rho) + S_{vn}(\rho) + P_{vn}(\rho) = \ln d. \label{eq:rel}
\end{align}  
Now, we can interpret $C_{re}(\rho)$ and $S_{vn}(\rho)$ in terms of quantum and classical uncertainties, respectively, i.e., $\mathcal{U}(\rho) := C_{re}(\rho) + S_{vn}(\rho)$\footnote{Here, we will not define $C_{re}(\rho) + S_{vn}(\rho)$ as $\mathcal{V}(\rho)$ because the first is not a variance.}. Following Ref. \cite{Korzekwa}, we can consider the dephasing map $D(\rho) = \sum_j \bra{j}\rho \ket{j} \ketbra{j}$. The projective measurements $\{ \ketbra{j} \}_{j = 1}^d$ related to the paths are a repeatable measurement, and so it is reasonable to demand that a second measurement should not reveal any quantum uncertainty in the state and is entirely classical. Thus, we can take $S(\rho||D(\rho)) = \min_{\iota \in I} S(\rho || \iota) =  C_{re}(\rho)$ as the quantum uncertainty and $S(\rho)$ as the classical uncertainty. If $\rho$ is pure, $S_{vn}(\rho) = 0$ and $\mathcal{U}(\rho) = C_{re}(\rho)$. On the other hand, if $[\rho, \ketbra{k}] = 0$, for some path index $k$, then $\rho$ is diagonal on the path basis. Hence $C_{re}(\rho) = 0$, since $S(\rho_{diag}) = S(\rho)$ and $\mathcal{U}(\rho) = S(\rho)$. Also, it is known that $C_{re}(\rho)$ is convex under classical mixtures \cite{Baumgratz} and $S(\rho)$ is concave under classical mixtures \cite{nielsen}. Hence, we can also interpret Eq. (\ref{eq:rel}) as a complete complementarity relation between uncertainties and predictability.
 
 Furthermore, we can use the fact that the coherences and the quantum correlations of $\rho$ give rise to quantum and classical uncertainties, respectively, to obtain a complete complementarity relation. In Ref. \cite{Maziero}, we obtained an incomplete complementarity relation using the $l_1$-norm as a measure of quantum coherence \cite{Baumgratz}, just by exploring the properties of the density matrix: since $\rho$ is positive semi-definite, we can use the fact that $\abs{\rho_{jk}} \le \sqrt{\rho_{jj}\rho_{kk}}$, $\forall j \neq k$ \cite{horn},  to obtain
 \begin{equation}
     C_{l_1}(\rho) \le \sum_{j \neq k} \sqrt{\rho_{jj} \rho_{kk}} \le d - 1
 \end{equation}
 which can be recast as a complementarity relation:
 \begin{align}
      C_{l_1}(\rho) +  P_{l_1}(\rho) \le d -1,
 \end{align}
 where $C_{l_1}(\rho) := \sum_{j \neq k} \abs{\rho_{jk}}$ is the $l_1$-norm quantum coherence, and $P_{l_1}(\rho) := d-1 - \sum_{j \neq k} \sqrt{\rho_{jj}\rho_{kk}}$ is a bona-fide measure of predictability. Now, we notice that
 \begin{equation}
     C_{l_1}(\rho) + P_{l_1}(\rho) = d - 1 + \sum_{j \neq k}(\abs{\rho_{jk}} - \sqrt{\rho_{jj}\rho_{kk}})
 \end{equation}
 can be rewritten as a CCR,
 \begin{equation}
     C_{l_1}(\rho) + W_{l_1}(\rho) + P_{l_1}(\rho) = d - 1,
 \end{equation}
 if we define $W_{l_1}(\rho) =  \sum_{j \neq k}(\sqrt{\rho_{jj}\rho_{kk}} - \abs{\rho_{jk}})$ as a measure of entanglement of the quanton, provided that the system is part of a bipartite pure quantum system. And we can see that $C_{l_1}(\rho)$ and $W_{l_1}(\rho)$ give rise to quantum and classical uncertainties by showing that these measures satisfy Luo's criteria: if $\rho$ is pure $\therefore$ $\abs{\rho_{jk}} = \sqrt{\rho_{jj}\rho_{kk}}$, $\forall j \neq k$, and $W_{l_1}(\rho) = 0$. In this case, we can interpret that $\rho$ is part of a bipartite pure separable quantum system. On the other hand, if $\rho$ is incoherent in the path basis, then $C_{l_1}(\rho) = 0$ and $W_{l_1}(\rho) =  \sum_{j \neq k}\sqrt{\rho_{jj}\rho_{kk}}$. For the extreme case $\rho = \sum_j \frac{1}{d}\ketbra{j}$, $W_{l_1}(\rho) = d - 1$ reaches its maximum. Now, the convexity of $C_{l_1}(\rho)$ was shown in Refs. \cite{Baumgratz, Maziero}. To show the concavity of $W_{l_1}(\rho)$, it is enough to note that $C_{l_1}(\rho) = \sum_{j \neq k} \abs{\rho_{jk}}$ and $f(x_1,..., x_d) = - \sum_{j\ne k}\sqrt{x_{j}x_{k}}$, with $x_j \in [0,1]$, are convex functions, therefore $-C_{l_1}(\rho)$  and $-f(x_1,..., x_d)  $ are concave \cite{Roberts}. Besides, it is worth mentioning that, in Ref. \cite{Leopoldo}, we succeed to show that $W_{l_1}$ is an entanglement monotone for bipartite pure cases, when restricted to the Schmidt's coefficients.

In addition, it is worth pointing out that Ref. \cite{Bera} obtained (incomplete) complementarity relations using the relative entropy and the $l_1$-norm quantum coherences as measures of visibility, while distinguishability measures were used for the `particleness' of the system. In contrast, in this work we use predicatibility measures for the particle aspect of the quanton. Besides, if the bipartite quantum system is mixed, the classical uncertainty measures can be taken as measures of the mixedness of the system, as in Refs. \cite{Dhar, Gamel}. For a recent discussion on this subject the reader is referred to Ref. \cite{Wang}.

To illustrate this complete complementarity relation, let us consider the state \cite{Giorgi} $\ket{\Phi(p,\epsilon)} = \sqrt{p \epsilon} \ket{0,0,0}_{A,B,C} + \sqrt{p(1 - \epsilon)}\ket{1,1,1}_{A,B,C} + \sqrt{(1 - p)/2}(\ket{1,1,0}_{A,B,C} + \ket{1,0,1}_{A,B,C})$, with $p, \epsilon \in [0,1]$. The reduced state $\rho_A$ is incoherent, meanwhile 
\begin{align}
\rho_B = \rho_C =& (p \epsilon + (1 - p)/2) |0 \rangle \langle 0 | + (p(1-\epsilon) + (1 - p)/2) |1 \rangle \langle 1 | \nonumber \\ 
& +(\sqrt{p(1-\epsilon)(1-p)/2)}|0 \rangle \langle 1 | + t.c.),
\label{eq:state}
\end{align}
where t.c. stands for the transpose conjugate. In Fig. \ref{fig:a} we plotted the coherence of $\rho_B, \rho_C$, as well as the correlation and predictability measures in Fig. \ref{fig:b} and \ref{fig:c}, respectively. By summing all three measures, we must saturate the complementarity relation,  as represented by the plane in Fig. \ref{fig:mes}, i.e., $C_{l_1}(\rho) + W_{l_1}(\rho) + P_{l_1}(\rho)$ is constant. 
\begin{figure}%
    \centering
    \subfloat[$C_{l_1}(\rho_B)$ as a function of $p,     \epsilon$.]{{\includegraphics[width=7cm]{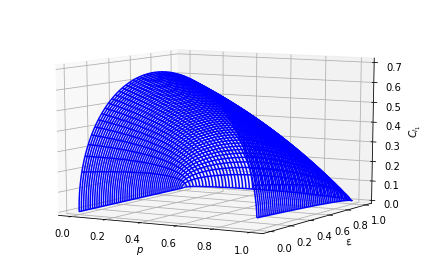}{\label{fig:a}} }}%
    \qquad
    \subfloat[$W_{l_1}(\rho_B)$ as a function of $p, \epsilon$.]{{\includegraphics[width=7cm]{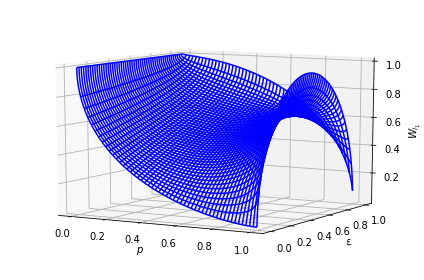}{\label{fig:b}} }}%
    \qquad
    \subfloat[$P_{l_1}(\rho_B)$ as a function of $p, \epsilon$.]{{\includegraphics[width=7cm]{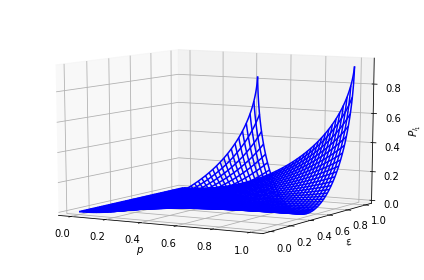}{\label{fig:c}} }}
    \qquad
    \subfloat[$C_{l_1}(\rho) + W_{l_1}(\rho) + P_{l_1}(\rho)$, as function of $p$ and $\epsilon$, is constant.]{{
    \includegraphics[width=7cm]{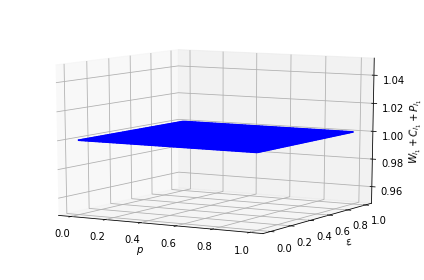}{\label{fig:mes}} }}
    \label{fig:example}
    \caption{(Color online) Quantum coherence and correlation and predictability measures, and their complementarity, for the state in Eq. (\ref{eq:state}).}
\end{figure}

Now, for a bipartite quantum system in the state $\ket{\Psi}_{A,B} = x\ket{0,1}_{A,B} + \sqrt{1 - x^2}\ket{1,0}_{A,B}$, with $x \in [0,1]$, we have
\begin{align}
    & W_{l_1}(\rho_A) = W_{l_1}(\rho_B) = 2x\sqrt{1 - x^2}, \label{ex1i} \\
    & P_{l_1}(\rho_A) = P_{l_1}(\rho_B) = 1 - 2x\sqrt{1 - x^2},\\
    & \mathcal{U}_c(\rho_A) = \mathcal{U}_c(\rho_B) = 2x^2(1 - x^2), \\
    & P_{l}(\rho_A) = P_{l}(\rho_B) = 1/2 - 2x^2(1 - x^2),\\
    & S_{vn}(\rho_A) = S_{vn}(\rho_B) = -x^2 \ln x^2 - (1 - x^2) \ln (1 - x^2), \\    
    & P_{vn}(\rho_A) = P_{vn}(\rho_B) = \ln 2 + x^2 \ln x^2 + (1 - x^2) \ln (1 - x^2), \label{ex1f}
\end{align}
where $\mathcal{U}_q(\rho_A) = \frac{1}{2}E^2(\Psi_{A,B})$, with $E(\Psi_{A,B})$ being the concurrence measure of entanglement \cite{Woot} and $W_{l_1}(\rho_A) = C^c_{l_1}(\rho_{A,B})$, with $C^c_{l_1}(\rho_{A,B})$ being the $l_1$-norm correlated coherence \cite{Tan}. In Fig. \ref{fig:meas}, we plotted the different measures of predictability and correlation for comparison.

\begin{figure}[t]
\includegraphics[scale=0.55]{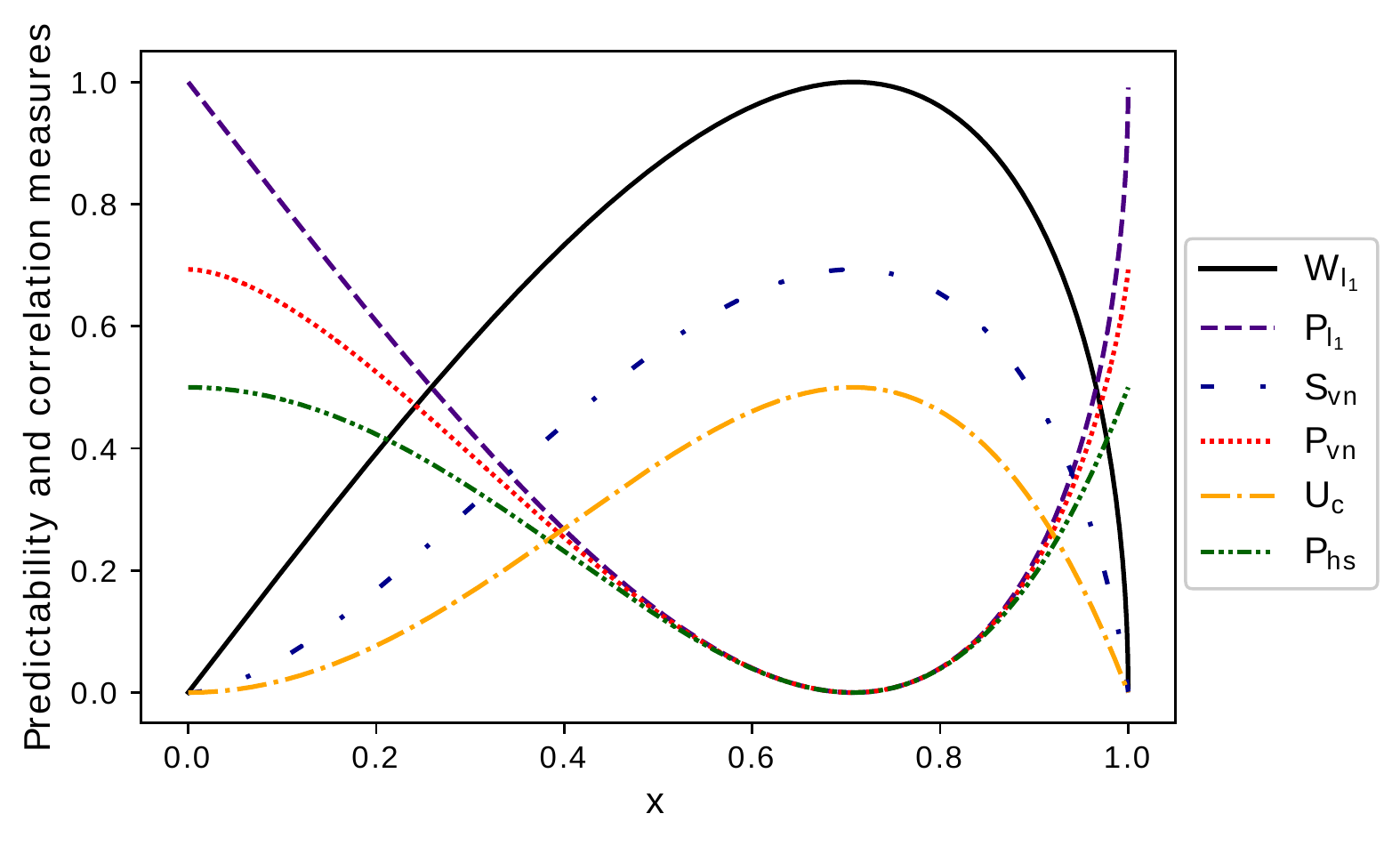}
\caption{(Color online) Comparison between different measures of predictability and their respective correlation measures of Eqs. (\ref{ex1i})-(\ref{ex1f}).}
\label{fig:meas}
\end{figure}

\subsection{The generalized Gell-Mann's matrices and its relation with complementarity}
From the variance of the generalized Gell-Mann's matrices (GMM), we will obtain the complete complementarity relation for the Hilbert-Schmidt norm (or $l_2$-norm) given by
\begin{align}
    P_l(\rho) + C_{hs}(\rho) + S_l(\rho) = \frac{d - 1}{d},
\end{align}
where $P_l(\rho) := S_l^{max} - S_l(\rho_{diag})$ is the predictability already defined in Sec. \ref{sec:comp}, $C_{hs}(\rho) := \sum_{j \neq k} \abs{\rho_{jk}}^2$ is the Hilbert-Schmidt quantum coherence and $S_l(\rho) = 1 - \Tr \rho^2$ is the linear entropy. By doing this, we generalize the relationship between wave-particle quantifiers and uncertainties explored in Refs. \cite{Bjork, Liu} for qubits. Let $\{\ket{j}\}_{j=1}^{d}$ be any given vector basis for $\mathbb{C}^{d}$. Using this basis, we can define the generalized
GMM as \cite{bertlmann}:
\begin{align}
\Gamma_{m}^{d} & :=  \sqrt{\frac{2}{m(m+1)}}\sum_{l=1}^{m+1}(-m)^{\delta_{l,m+1}}\ketbra{l},\\
\Gamma_{j,k}^{s} & :=  \ketbra{j}{k} + \ketbra{k}{j},\\
\Gamma_{j,k}^{a} & :=  -i(\ketbra{j}{k} - \ketbra{k}{j}),
\end{align}
where, if not stated otherwise, we use the following possible values for the indexes $m,j,k$: $ m=1,\cdots,d-1\mbox{ and }1\le j<k\le d$. For $d = 2$, the generalized Gell-Mann's matrices reduce to the well known Pauli matrices. Besides, one can easily see that these matrices are Hermitian and traceless. If we use $\Gamma_{0}^{d}$ for the $d \times d$ identity matrix, it is not difficult to see that under the Hilbert-Schmidt's inner product,
\begin{equation}
\langle A|B\rangle_{hs}:=\mathrm{Tr}(A^{\dagger}B),
\end{equation}
with $A,B\in\mathbb{C}^{d \times d}$, the set
\begin{equation}
\left\{ \frac{\Gamma_{0}^{d}}{\sqrt{d}},\frac{\Gamma_{m}^{d}}{\sqrt{2}},\frac{\Gamma_{j,k}^{\tau}}{\sqrt{2}}\right\} ,
\end{equation}
with $\tau=s,a$, forms an orthonormal basis for $\mathbb{C}^{d \times d}$
\cite{bertlmann}. So, any matrix $X\in\mathbb{C}^{d \times d}$
can be decomposed in this basis. In particular, we can decompose the density operator as follows
\begin{equation}
\rho =\frac{1}{d}\mathrm{Tr}(\rho)\Gamma_{0}^{d}+\frac{1}{2}\sum_{j}\langle\Gamma_{j}^{d}|\rho\rangle\Gamma_{j}^{d}+\frac{1}{2}\sum_{k,l,\tau}\langle\Gamma_{k,l}^{\tau}|\rho\rangle\Gamma_{k,l}^{\tau}.
\end{equation}
Now, since
\begin{align}
    & \sum_{m = 1}^{d - 1} \expval{\Gamma^d_m}^2 = 2\big(\sum_j \rho^2_{jj} - 1/d\big) = 2 P_l (\rho),\label{eq:varp}  \\
    & \sum_{j<k}\big(\expval{\Gamma^s_{j,k}}^2 + \expval{\Gamma^a_{j,k}}^2\big) = 2 \sum_{j \neq k}\abs{\rho_{jk}}^2 = 2 C_{hs}(\rho), \label{eq:varc} 
\end{align}
where $C_{hs}(\rho) = \sum_{j \neq k}\abs{\rho_{jk}}^2$ is the Hilbert-Schmidt (or $l_2$-norm) quantum coherence \cite{Jonas}, and $\expval{\Gamma} = \Tr (\rho \Gamma)$. The sum of the variances of these subsets of observables is given as follows:
\begin{align}
 & \sum_{m} \mathcal{V}(\rho,\Gamma^d_m) = \frac{2(d -1)}{d} - 2 P_l(\rho), \label{eq:gdia}\\
 & \sum_{j<k}\Big(\mathcal{V}(\rho,\Gamma^s_{j,k}) + \mathcal{V}(\rho,\Gamma^a_{j,k})\Big) = 2(d - 1) - 2 C_{hs}(\rho). \label{eq:goff}
\end{align}
By summing Eqs. (\ref{eq:gdia}) and (\ref{eq:goff}), we obtain the complete complementarity relation
\begin{equation}
P_l(\rho) + C_{hs}(\rho) + \mathcal{C}(\rho, \Gamma)= \frac{d -1}{d}. \label{eq:ccr}
\end{equation}
where $ \mathcal{C}(\rho, \Gamma):= \frac{1}{2}\sum_{m} \mathcal{V}(\rho,\Gamma^d_m) + \frac{1}{2}\sum_{j<k}(\mathcal{V}(\rho,\Gamma^s_{j,k}) + \mathcal{V}(\rho,\Gamma^a_{j,k})) - (d - 1)$ is a measure of classical uncertainty, once it satisfies Luo's criteria. In addition, Eq. (\ref{eq:ccr}) is equivalent to the complete complementarity relation obtained by us in Ref. \cite{Marcos}, exploring the purity of multipartite quantum systems, where $\mathcal{C}(\rho, \Gamma) = 1 - Tr \rho^2 =  S_l$ is measuring the quantum correlations of $\rho$ with other systems, provided that the quanton is part of a multipartite pure quantum system. Besides, for bipartite pure quantum systems,  the CCR in Eq. (\ref{eq:ccr}) is equivalent to that obtained by Jakob and Bergou \cite{Jakob} using the concurrence as a measure of quantum correlation. As $\rho$ represents a mixed quantum state in general, $P_l(\rho) + C_{hs}(\rho) \le \frac{d-1}{d}$ and we have the following uncertainty relation for the generalized Gell-Mann's matrices:
\begin{equation}
\sum_{m} \mathcal{V}(\rho,\Gamma^d_m) + \sum_{j<k}(\mathcal{V}(\rho,\Gamma^s_{j,k}) + \mathcal{V}(\rho,\Gamma^a_{j,k})) \ge 2(d - 1). \label{eq:uncgell}
\end{equation}
For instance, for $d = 2$, the generalized Gell-Mann's matrices reduce to the well known Pauli matrices $(\sigma_x, \sigma_y, \sigma_z)$, and Eq. (\ref{eq:uncgell}) reduces to
\begin{align}
    \mathcal{V}(\rho, \sigma_x) + \mathcal{V}(\rho, \sigma_y) + \mathcal{V}(\rho, \sigma_z) \ge 2.
\end{align}
While, from Eqs. (\ref{eq:varp}) and (\ref{eq:varc}) we have a trade-off between the variances of the Pauli matrices and the measures that quantify complementarity:
\begin{align}
    \mathcal{V}(\rho, \sigma_x) + \mathcal{V}(\rho, \sigma_y) + \mathcal{V}(\rho, \sigma_z) = 3 - 2(C_{hs}(\rho) + P_l(\rho)),
\end{align}
which is a equivalent to the equation (8) of Ref. \cite{Liu}. One can see this by noting that $P_l = \frac{1}{2}P^2$ and $C_{hs} = \frac{1}{2}V^2$, where $P,V$ are the well known predictability and visibility measures used in Refs. \cite{Yasin, Liu}. The only difference is that the authors in Ref. \cite{Liu} consider the visibility operator as $\hat{V} = \cos \phi \sigma_x + \sin \phi \sigma_y$ and the predictability operator as $\hat{P} = \sigma_z$ such that the trade-off relation adds up to $2$: $P^2 + V^2 + \mathcal{V}(\hat{V}) + \mathcal{V}(\hat{P}) = 2$. Therefore, one can see that Eq. (\ref{eq:ccr}) generalizes such relationship between complementarity measures and uncertainties for $d$-dimensional quantum systems.

\section{Conclusions}
\label{sec:conc}
The relationships between Luo's criteria for quantum and classical uncertainties and D\"urr-Englert et al.'s criteria for wave-particle duality together with the criteria for entanglement measures were discussed. This lead naturally to the notion that quantum entanglement gives rise to local classical uncertainties, provided that the quanton is part of a pure bipartite quantum system, while quantum coherence gives rise to quantum uncertainties. In addition, we showed that the quantum uncertainty of all $d$-paths is equivalent to the Wigner-Yanase quantum coherence, meanwhile the classical uncertainty can be taken as a correlation quantifier. By exploring the relation between uncertainties and wave-particle duality, we obtained quantum correlation measures completing the $l_1$-norm and $l_2$-norm complementarity relations, that can also be taken as classical uncertainty measures. Therefore, our work connects two crucial concepts that were guidelines for the development of Quantum Mechanics, which have been recently formalized: the separation of the total uncertainty in its classical and quantum parts developed by Luo; and the quantification of complementarity made by several authors, as emphasized in Sec. \ref{sec:intro}. 

Besides, we believe that our results will help in the realization that neither complementarity nor uncertainty can be considered a more fundamental aspect of quantum theory, what is still a vivid debate \cite{Walther, Storey, Rempe, Vathsan, Taqu}. Instead, complementarity and uncertainty are intrinsically connected to each other, for both follow directly from the mathematical structure of quantum mechanics. Lastly, to summarize the role of the different entropies in uncertainty and complementarity relations of multipartite pure quantum states explored in this paper, we use Table (\ref{tab:entro}).
\begin{widetext}
\begin{table}[htp]
\centering
\vspace{0.5cm}
\begin{tabular}{l|l|l|}
 Entropy & Usefulness in uncertainty relations & Usefulness in complementarity relations \\
 \hline
 $S_l(\rho)$ & Measure of classical uncertainty  & Measure of correlation with other systems \\
 $S_l(\rho_{diag})$ & Measure of total uncertainty & Can be used to define predictability  \\
 $S_{vn}(\rho)$ & Measure of classical uncertainty & Measure of correlation with other systems \\
 $S_{vn}(\rho_{diag})$ & Measure of total uncertainty & Can be used to define predictability 
\end{tabular}
\caption{Role of different entropies in uncertainty and complementarity relations.}
\label{tab:entro}
\end{table}
\end{widetext}
Lastly,  in this work we noticed that the overall distinction between classical, total, and quantum uncertainty is that classical uncertainty is characterized from the eigenvalues of the density matrix, the total uncertainty is characterized from the diagonal elements of the density matrix (i.e., the probability distribution acquired in a experiment) and the quantum uncertainty is defined as the difference between the two using whichever measure of uncertainty is most convenient.
  
\begin{acknowledgments}
This work was supported by the Coordena\c{c}\~ao de Aperfei\c{c}oamento de Pessoal de N\'ivel Superior (CAPES), process 88882.427924/2019-01, and by the Instituto Nacional de Ci\^encia e Tecnologia de Informa\c{c}\~ao Qu\^antica (INCT-IQ), process 465469/2014-0.
\end{acknowledgments}

\end{document}